\documentclass[11pt,reqno]{article}
\usepackage{amsfonts,amsmath,amssymb,amsthm}

\usepackage[titletoc]{appendix}
\usepackage{bbold}
\usepackage{comment}

\usepackage[shortlabels]{enumitem}
  \setlist[enumerate,1]{leftmargin=15pt}
  \setlist[itemize,1]{leftmargin=10pt}
  \setlist[description,1]{leftmargin=15pt}

\usepackage[text={6in,9.5in},centering]{geometry}
\usepackage{setspace}
\usepackage{url}
\usepackage[dvipsnames]{xcolor}

\usepackage[bookmarks]{hyperref}

\newenvironment{Red}{\noindent\color{red}}{}

\newtheorem{theorem}{Theorem}[section]

\newtheorem{corollary}[theorem]{Corollary}

\newtheorem{lemma}[theorem]{Lemma}
\newtheorem{proposition}[theorem]{Proposition}

\theoremstyle{definition}

\newtheorem{convention}[theorem]{Convention}
\newtheorem{definition}[theorem]{Definition}

\newtheorem{setting}{Setting}

\newenvironment{lquote}
  {\list{}{\leftmargin=1.5em\rightmargin=1em}\item[]}%
  {\endlist}

\newcommand\A{\ensuremath{\mathcal A}}
\renewcommand\a{\ensuremath{|a\rangle\langle a|}}
\newcommand\ang[1]{\ensuremath{\big\langle#1\big\rangle}}
\renewcommand\b{\ensuremath{\beta}}
\newcommand\bra[1]{\ensuremath{\langle#1|}}
\newcommand\braket[2]{\ensuremath{\langle#1\,|\,#2\rangle}}
\newcommand\C{\ensuremath{\mathcal C}}
\newcommand\CS[1]{\ensuremath{\mathrm{CS}(#1)}}
\newcommand\DO[1]{\ensuremath{\mathrm{DO(#1)}}}
\newcommand{\dg}{^\dag}
\newcommand\End[1]{\ensuremath{\mathrm{End}(#1)}}
\newcommand\G{\ensuremath{\mathcal G}}
\renewcommand{\H}{\ensuremath{\mathcal H}}
\newcommand\Id{\ensuremath{\mathbb 1}}
\newcommand\iset[1]{\ensuremath{\left\langle#1\right\rangle}}
\newcommand{\K}{\ensuremath{\mathcal K}}
\newcommand{\ket}[1]{\ensuremath{|#1\rangle}}
\newcommand{\ketbra}[2]{\ensuremath{|#1\rangle\langle#2|}}
\renewcommand\L{\ensuremath{\mathcal L}}
\newcommand\norm[1]{\ensuremath{\Vert#1\Vert}}
\newcommand\Out{\ensuremath{\mathrm{Output}}}
\newcommand\ox{\ensuremath{\otimes}}
\newcommand\pbrz{\ensuremath{P_\T(\b\,| \rz)\,}}
\newcommand\pbs{\ensuremath{P(\b | \sigma)}}
\renewcommand\phi{\varphi}
\newcommand\qef{\hfill$\triangleleft$} 
\renewcommand\r{\ensuremath{\rho}}
\newcommand\rz{\ensuremath{\,\mathring\rho\,}}
\newcommand\s{\ensuremath{\Sigma}}

\newcommand\T{\ensuremath{\mathcal T}}
\newcommand{\Tr}{\ensuremath{\mathrm{Tr}}}

\title{Input independence}
\author{Yuri Gurevich and Andreas Blass\\
\normalsize University of Michigan, Ann Arbor, Michigan, USA}
\date{}

\begin{document}
\maketitle

\footnotetext{The authors are partially supported by the US Army Research Office under W911NF-20-1-0297.}
\thispagestyle{empty}

\begin{abstract}
We establish the following \emph{input independence} principle.
If a quantum circuit \C\ computes a unitary transformation $U_\mu$ along a computation path $\mu$, then the probability that computation of \C\ follows path $\mu$ is independent of the input.
\end{abstract}
\maketitle
\thispagestyle{empty}

\section{Introduction} 
\label{s:intro}

In the analysis of quantum computations, estimating the probability that a given circuit computes the desired output is rather important.
Typically, the circuit is good enough if that probability is high enough.
Static analysis of the circuit may reveal good computation paths, those leading to the desired outcome.
But estimating the probability $P$ that the computation follows a good path may not be easy.
A priori $P$ depends on the input.

The main result of this paper is the \emph{input independence} principle: If a quantum circuit \C\ computes a unitary transformation $U_\mu$ along a computation path $\mu$, then the probability that computation of\, \C\ follows path $\mu$ is independent of the input.
The principle implies that, if the desired output is given by a unitary transformation $U$, then the probability $P$ that computation of \C\ follows a path leading to the desired outcome does not depend on the input.
Furthermore, if the unitary transformation is allowed to depend on the path, then $P$ still does not depend on the input.

First, based on the syntax and semantics of quantum circuits in \cite{G250},  we make the principle mathematically precise.
In the process, we recall terminology from \cite{G250} and repeat some, though not all, of the definitions.

\begin{setting}\label{s1}\mbox{}
\begin{enumerate}[A.]
\item Following the standard textbook \cite{NC} of Michael Nielsen and Isaac Chuang, we define a \emph{(general) measurement} over a Hilbert space \K\ as an indexed family $M = \ang{L_i: i\in I}$ of (bounded linear)  operators on \K\ where $\sum_{i\in I} L_i^\dag L_i$ is the identity operator $\Id$ on \K; elements of the index set $I$ are \emph{(classical) outcomes} of $M$ \cite[\S2.2.3]{NC}.
    However, contrary to \cite{NC}, we do not presume that \K\ is finite dimensional or that $I$ is finite.
    Instead, as in \cite{BLM}, we presume that \K\ is separable (and complex), $I$ is countable, and the sum converges in the weak operator topology%
    \footnote{We began this work in the context of quantum computing and required that Hilbert spaces are  finite dimensional and measurements have only finitely many outcomes.
    However, the results work just as well in the broader setting adopted here.}.
    A \emph{state} in a Hilbert space \K\ is given by a density operator on \K; the set of these density operators is denoted DO(\K).

\item In much of the literature, the meaning of ``quantum circuit'' is ambiguous.
    A circuit typically contains unitary operations, often contains measurements, and may contain classical channels.
   The ambiguity is whether the circuit also tells which gates fire simultaneously and, for those that don't, the order in which they fire.
   In the present paper, a quantum circuit may contain unitary operations, measurements, and classical channels. And by default a circuit comes together with a fixed execution schedule $(B_1, B_2, \dots, B_N)$ where $B_1$ is the set of gates to be fired first, $B_2$ is the set of gates to be fired next, and so on%
   \footnote{In \cite{G250}, we distinguished between circuits with execution schedules and without.
   A circuit with a fixed execution schedule was called a \emph{circuit algorithm} there.}.
   It is required of course that the sets $B_n$ are disjoint, that their union contains all the gates, and that, for every $n$ and every gate $G\in B_n$, the set $\bigcup_{k<n}B_k$  of gates scheduled to fire before $G$ contains all (quantum and classical) prerequisites of $G$.

\item For uniformity, we treat unitary operators as measurements with a single outcome, so that all gates are measurement gates.
    Every gate $G$ in a quantum circuit \C\ is assigned a finite nonempty set of measurements, called $G$ \emph{measurements}.
    It is assumed without loss of generality that different $G$ measurements have disjoint outcome sets.
    In addition, $G$ is assigned a \emph{selection} function $\Sigma_G$ that, in runtime, picks the $G$ measurement to be executed.
    Let \CS{G} be the set of \emph{classical sources} of $G$, i.e.\ the gates with classical channels to $G$.
    Given outcomes $o_F$ of the measurements performed by the classical sources $F$ of $G$, the $G$ measurement $\Sigma_G \iset{o_F: F\in \CS{G}}$ is picked.
    Without loss of generality, we presume that there are no superfluous $G$ measurements: every $G$ measurement is in the range of $\s_G$.
    In particular, if $G$ has no classical sources, then it is assigned (a set comprising) a  single $G$ measurement.

\item A \emph{(computation) path}\footnote{Paths were called tracks in \cite{G250}} $\mu$ through a quantum circuit \C\ is an assignment to each gate $G$ of an outcome $\mu(G)$ of the $G$-measurement $M_\mu(G) = \Sigma_G\iset{\mu(F): F\in \CS{G}}$.
    Note that this last equation is a coherence requirement on the measurements performed along the path.
    Given an input for \C, a path $\mu$ represents a potential computation of \C\ which is said to \emph{follow} path $\mu$ on the given input.

\item Consider quantum circuits \C\ with principal inputs and outputs (being states) in a Hilbert state \H\ of dimension $\ge2$.
    \C\ may also use ancillas (initialized in a fixed pure state) and may produce garbage to be discarded at the end of the computation.
    Let \A\ be a Hilbert space hosting the ancillas (which we may view as a single higher-dimensional ancilla) at the beginning of computation and the garbage at the end.
    Accordingly, a full input of \C\ has the form $\rz = \r\ox\a$ where $\r\in \DO{\H}$ and \ket a is a fixed unit vector in \A.
    If the computation of \C\ on input \rz\ follows $\mu$ then the output is a density operator $\Out_\C(\mu\, | \rz)$ on \H\ox\A, and the principal output is the density operator on \H\ given by the partial trace $\Tr_\A\big(\Out_\C(\mu, | \rz)\big)$. \qef
\end{enumerate}
\end{setting}

\begin{theorem}[Input independence]\label{t:ii}
Suppose that a quantum circuit \C\ computes a unitary operator $U_\mu: \H\to\H$ along a computation path $\mu$ in the sense that density operators $\Tr_\A\big( \Out_\C (\mu\, |\rz) \big)$ and $U_\mu\rho\, U_\mu\dg$ represent the same state in \H\ for all $\r\in \DO{\H}$.
Then the probability $P_\C(\mu\, |\rz)$ that the computation of\, \C\ follows path $\mu$ on principal input \r\ does not depend on \r.
\end{theorem}

The proof of the theorem gives a tiny bit more:
It suffices to assume that \C\ computes $U_\mu$ along $\mu$ just on pure inputs.

To simplify notation, we assumed in Setting~\ref{s1}.E that principal inputs and outputs of \C\ are in the same Hilbert space and that the ancillas and garbage are in the same Hilbert space.
Accordingly, in Theorem~\ref{t:ii}, \C\ computes a unitary operator (rather than unitary transformation) along path $\mu$.
The generalization where the two assumptions are dropped is obvious.

\begin{corollary}\label{c:iimult}
Let \G\ be a set of computation paths through a quantum circuit \C.
Suppose that \C\ computes a unitary operator along every computation path in \G, possibly different operators along different paths.
Then the probability that a computation of \C\ follows a path in \G\ does not depend on the input.
\end{corollary}

The \G\ in the corollary can represent the set of good paths mentioned in the motivating paragraphs.

\subsubsection*{Related work}
In 2017, Vadym Kliuchnikov conjectured that, if a quantum circuit computes a unitary operator (the same unitary operator along every computation path), then the probability that its computation follows a particular path does not depend on the input.
His conjecture provoked this investigation.
We confirm the conjecture.

\begin{corollary}\label{c:iiall}
Suppose that a quantum circuit \C\ computes a unitary operator $U$ in the sense that \C\ computes $U$ along every computation path.
Then the probability that a computation of \C\ follows a particular path does not depend on the input.
\end{corollary}

Paul Busch formulated the ``no information without disturbance'' principle for quantum measurements \cite[Theorem~2]{Busch} and sketched the proof of the principle; a slightly more informative version of his sketch appeared earlier in a book by Paul Busch, Pekka J. Lahti, and Peter Mittelstaedt \cite[p.~32]{BLM}.
Busch's principle is a special case of Corollary~\ref{c:iiall} where \C\ is just a single measurement and $U$ the identity operator.
(Busch works with positive operator-valued measurements, POVMs, rather than general measurements but this is not important for our story.
``POVMs are best viewed as a special case of the general measurement formalism, providing the simplest means by which one can study general measurement statistics, without the necessity for knowing the post-measurement state'' \cite[\S2.2.6, Box~2.5]{NC}.
Under our broader assumptions in Setting~\ref{s1}, POVMs still can be viewed as a special case of the general measurement formalism.)

\subsubsection*{Organization of this paper}

The input independence theorem is proved in the rest of this paper.
In order to analyze circuit computation, it is useful to untangle and
separate distinct computation paths.
We do that in \S\ref{s:reduce} where we introduce measurement trees whose branches can be viewed as their computation paths.
The reduction theorem, Theorem~\ref{t:reduction}, asserts that for
every quantum circuit \C\ there exist a measurement tree \T\ and a
one-to-one correspondence $\mu \mapsto \b_\mu$ from the computation
paths of \C\ to the branches of \T\ such that the probability that
computation of \C\ follows $\mu$ on a given (full) input is exactly
the probability that the computation of \T\ follows branch $\b_\mu$ on that input and the two computations produce the same output.

After investigating some basic properties of measurement trees in \S\ref{s:trees}, we prove the input independence theorem for measurement trees, Theorem~\ref{t:iit}, in \S\ref{s:ii}.
By the reduction theorem, the input independence theorem for circuits follows from that for measurement trees.

Appendix~A establishes a function constancy criterion which is used in \S\ref{s:ii} and which is arguably of independent interest.
Relegating that task to an appendix makes the exposition in \S\ref{s:ii} cleaner.
Appendix~B is devoted to a generalization of Theorem~\ref{t:ii}.

\section{Reduction to measurement trees}\label{s:reduce}

As usual we take density operators over a Hilbert space \K\ to be nonzero, positive semidefinite, Hermitian operators on \K.
However, instead of requiring them to have trace 1, we only require them to have finite traces, i.e., to be trace-class operators.
The set of density operators over \K\ will be denoted \DO{\K}.
A density operator is \emph{normalized} if its trace is 1.

\begin{convention}\label{con}
We use density operators, not necessarily normalized, to represent (possibly mixed) states over \K. A density operator $\rho$ represents the same state as its normalized version $\rho/\Tr(\rho)$.
If a measurement $M = \ang{L_j: j\in J}$ in the state (represented by) $\rho$ produces outcome $j$, then we will usually use density operator $L_j\rho L_j\dg$ (rather than any scalar multiple of it) to represent the post-measurement state.
Notice that $\rho\mapsto L_j\rho L_j\dg$ is a linear operator and the probability of outcome $j$ is $\Tr(L_j\rho L_j\dg)/\Tr(\rho)$. \qef
\end{convention}

To analyze computations of a quantum circuit, it is convenient to represent the circuit by a tree where the computations are represented by branches.

\begin{definition}\label{d:meas tree}
A \emph{measurement tree}, a \emph{meas tree} in short, over a Hilbert space \K\ is a directed rooted tree with a finite bound on the lengths of branches, where every non-leaf node $x$ has a measurement $M_x$ over \K\ associated to it, and the edges emanating from $x$ are labeled in one-to-one fashion with the classical outcomes of $M_x$. \qef
\end{definition}

All our results about meas trees would remain true if the bounded-branch-length assumption were replaced by a weaker assumption that every branch is of finite length.
But the bounded-branch-length trees suffice to support the reduction theory, Theorem~\ref{t:reduction}, below.

Let $\T$ be a meas tree. A \emph{route} on \T\ is a sequence of edges $(x_0,x_1), (x_1,x_2), \dots, (x_{n-1},x_n)$ where $x_0$ is the root;
it may be given by the labels of those edges: $(o_1,\dots,o_n)$ where $o_{i+1}$ is an outcome of the measurement $M_{x_i}$.
The same notation $(o_1,\dots,o_n)$ may be used for the final node $x_n$ of the route; to avoid ambiguity, we may also use notation \End{o_1,\dots,o_n} for $x_n$.
A route is a \emph{branch} if its final node is a leaf.

View \T\ as the algorithm which works as follows on input $\sigma_0\in\DO{\K}$.
If the root is a leaf (and thus the only node), do nothing.
Otherwise, perform the root's measurement in state $\sigma_0$ producing,  with some probability $p_1$, an outcome $o_1$ and post-measurement state $\sigma_1$. If node \End{o_1} is not a leaf, perform its measurement in state $\sigma_1$ producing,  with some probability $p_2$, an outcome $o_2$ and post-measurement state $\sigma_2$, and so on until a leaf  is reached.
Let \b\ be the resulting branch $(o_1, o_2, \dots, o_N)$.
The (quantum) output $\Out_\T(\b |\sigma_0)$ of this computation is the post-measurement state $\sigma_N$ resulting from the last measurement performed.
The probability $P_\T(\b |\sigma_0)$ that computation of \T\ follows
branch \b\ on input $\sigma_0$ is $p_1\!\cdot p_2\cdots p_N$.

Consider quantum circuits \C\ with (full) inputs and outputs in some Hilbert space \K.
Assume Setting~\ref{s1} (so that $\K = \H \ox \A$ but the internal structure of \K\ will play no role in this section.)
Without loss of generality, we may assume that \C\ comes with a fixed linear order of its gates.

A \emph{bout} of gates of \C\ is a nonempty set $B$ of gates where no gate is a (quantum or  classical) prerequisite for another gate; thus all $B$ gates could be executed in parallel after all their prerequisites have been executed.

A \emph{schedule} of the circuit is a sequence
$(B_1, B_2, B_3, \dots, B_N)$
of disjoint gate bouts such that every gate set
$\bigcup_{i\le n} B_i$
is closed under prerequisites and $\bigcup_{n\le N} B_n$ contains all the gates.

A schedule is \emph{linear} if all its bouts are (composed of) single gates.
Paths through a circuit with linear schedule $(G_1, G_2, \dots, G_N)$ can be represented by sequences $(o_1, o_2, \dots, o_N)$ where each $o_n$ is an outcome of the $G_n$ measurement $\Sigma_{G_n}\iset{o_i: G_i \in \CS{G_n}}$.

\begin{lemma}\label{l:linear}
Let \C\ be a quantum circuit with a linear schedule $(G_1, G_2, \dots, G_N)$ with (full) inputs and outputs in Hilbert space \K.
There exists a meas tree \T\ over \K\ such that for every path $\mu$ through circuit \C\ and every $\sigma$ in \DO{\K} we have:
\begin{enumerate}
\item $\b_\mu = (\mu(G_1), \mu(G_2), \dots, \mu(G_N))$ is a branch of \T, and every branch of \T\ is obtained that way,
\item $\Out_\C(\mu|\sigma) = \Out_\T(\b_\mu | \sigma)$, and
\item $P_\C(\mu|\sigma) = P_\T(\b_\mu | \sigma)$.
\end{enumerate}
\end{lemma}

\begin{proof}
We construct the desired meas tree \T.
The nodes of \T are initial segments  $(o_1, \dots, o_n)$ of paths through \C.
If $n<N$ then segments $(o_1, \dots, o_n, o_{n+1})$ are children of
the segment $(o_1, \dots, o_n)$, the label of the edge from segment
$(o_1, \dots, o_n)$ to segment $(o_1, \dots, o_n, o_{n+1})$ is
$o_{n+1}$, and the measurement assigned to segment $(o_1, \dots, o_n)$
is $\Sigma_{G_{n+1}}\!\iset{o_i: G_i \in \CS{G_{n+1}}}$.
In particular, the empty segment is assigned the unique measurement of the first gate $G_1$.
As in the definition of computation paths in Setting~\ref{s1},
this description of children incorporates the coherence requirement for the measurements along the route $(o_1, \dots, o_{n+1})$.
All three claims are easy to verify from the definition of \T.
\end{proof}

Given measurements $M_1 = \iset{A_i: i\in I}$ and $M_2 = \iset{B_j:
  j\in J}$ over Hilbert spaces $\K_1, \K_2$ respectively, consider the
tensor product $M_1 \ox M_2$, i.e.\ the measurement \iset{A_i\ox B_j:
(i,j)\in I\times J} over $\K_1 \ox \K_2$.
If the probabilities of outcomes $i,j$ on inputs $\sigma_1, \sigma_2$ are $p_i, q_j$ respectively, then the probability of outcome $(i,j)$ on input $\sigma_1\ox\sigma_2$ is $p\cdot q$.
The tensor products of more than two (but finitely many) measurements are defined similarly.

The following theorem extends Lemma~\ref{l:linear} to the case of schedules which may not be linear.

\begin{theorem}[Reduction]\label{t:reduction}
Let \C\ be a quantum circuit with (full) inputs and outputs in Hilbert space \K.
There exist a meas tree \T\ over \K\ and a one-to-one correspondence
$\mu \mapsto \b_\mu$ from the set of paths through \C\ onto the set of
branches of \T\ such that for every path $\mu$ through circuit \C\ and
every $\sigma$ in \DO{\K} we have:
\begin{enumerate}
\item $\Out_\C(\mu|\sigma) = \Out_\T(\b_\mu | \sigma)$, and
\item $P_\C(\mu|\sigma) = P_\T(\b_\mu | \sigma)$.
\end{enumerate}
\end{theorem}

\begin{proof}
Let $(B_1, B_2, \dots, B_N)$ be the schedule of \C.
By Lemma~\ref{l:linear}, it suffices to construct a circuit $\C'$ with
linear schedule admitting a one-to-one correspondence $\mu \mapsto
\mu'$ from the paths through \C\ onto those of $\C'$ so that $\Out_\C(\mu|\sigma) = \Out_\T(\b_\mu | \sigma)$ and
$P_\C(\mu|\sigma) = P_{\C'}(\mu'|\sigma)$ for all $\sigma$ in \DO{\K}
and all paths $\mu$ through \C.

We construct the desired circuit $\C'$.
The gates of $\C'$ are the bouts $B_n$ of the schedule of \C.
If the \C\ gates of $B_n$ are $F_1 < F_2 < \dots$ (in the fixed order of \C\ gates), then $B_n$ measurements are tensor products $M_{F_1} \ox M_{F_2} \ox\dots$ where each $M_{F_j}$ is an $F_j$ measurement in \C.
Since, for every gate $F$ in \C, the outcomes of different $F$ measurements are disjoint, the outcomes of different $B_n$ measurements are disjoint.

Next we define a selection functions $\Sigma'_n = \Sigma'_{B_n}$ for each $n = 1, \dots, N$.
A $\C'$ gate $B_i$ is a classical source for a $\C'$ gate $B_n$ if some \C\ gate in $B_i$ is a classical source for a \C\ gate in $B_n$.
If the \C\ gates of $B_i$ are $E_1 < E_2 < \dots$ and, in runtime, they produce outcomes $o_1, o_2, \dots$ respectively, then the outcome $O_i = (o_1, o_2, \dots)$ of $B_i$ is sent to $B_n$.
Note that $O_i$ determines the outcomes of all \C\ gates in $B_i$.
It follows that, for every \C\ gate $F$ in $B_n$, the outcomes $O_i$ determine outcomes of all classical sources $E$ of $F
$ in \C, and therefore determine an $F$ measurement $M_F$ chosen by the selection function $\Sigma_F$ of \C.
If the \C\ gates of $B_n$ are $F_1 < F_2 < \dots$, then
$\Sigma'_n\iset{O_i: B_i \in \CS{B_N}} = M_{F_1} \ox M_{F_2} \ox\dots$.

Finally, for every path $\mu$ through \C, we define a path $\mu'$ through $\C'$.
For each $n = 1,\dots, N$, we need to define an outcome $\mu'(B_n)$.
If the \C\ gates of $B_n$ are $F_1 < F_2 < \dots$, then $\mu'(B_n) = (\mu(F_1), \mu(F_2), \dots)$.
It is easy to see that every path through $\C'$ is obtained this way and that $\Out_\C(\mu|\sigma) = \Out_\T(\b_\mu | \sigma)$ and
$P_\C(\mu|\sigma) = P_{\C'}(\mu'|\sigma)$ for all $\sigma$ in \DO{\K}
and all paths $\mu$ through \C.
\end{proof}

\section{Cumulative operators}
\label{s:trees}

Let \T\ be a meas tree over a Hilbert space \K.
A branch \b\ of \T\ is \emph{attainable on} input $\sigma$ in DO(\K) if the probability $P(\b |\sigma)$ that computation of \T\ on input $\sigma$ follows branch \b\ is strictly positive.

\begin{definition}
For every branch $\b = (o_1, o_2, \dots, o_N)$ of \T, the \emph{cumulative operator} $C_\b$ is the composition
\[ C_\b = A_N \circ A_{N-1} \circ \cdots \circ A_2 \circ A_1 \]
where each $A_{n+1}$ is the operator of the measurement at node \End{o_1,\dots, o_n} that produces the outcome $o_{n+1}$. In the special case of the one-node tree, the length-zero branch has $C_\b=\Id$. \qef
\end{definition}

Notice that, for every input $\sigma$ in DO(\K), $C_\b\sigma C_\b\dg$ is exactly $\Out_\T(\b | \sigma)$.

\begin{lemma}\label{l:cum}
\begin{enumerate}\mbox{}
\item $\sum_\b C_\b\dg C_\b = \Id$,
\item $\displaystyle\pbs = \frac{\Tr(C_\b\sigma C_\b\dg)}{\Tr(\sigma)}$,
\item $\pbs=0 \iff  C_\b\sigma C_\b\dg=0$.
\end{enumerate}
\end{lemma}

\begin{proof}
Induction on the number of non-leaf nodes in \T.
\end{proof}

By the first two claims of the lemma, the indexed family of operators $C_\beta$ a measurement (that could be called the \emph{aggregate measurement} of \T\ in terms of \cite{G250}).

\begin{setting}\label{s2}
Let \A, \H, and \ket a be as in Setting~\ref{s1}.E.
Consider a meas tree \T\ with principal inputs and outputs in DO(\H).
Like circuits, \T\ may use ancillas and may produce garbage; the full input of \T\ has the form $\r\ox \a$ where \r\ ranges over DO(\H) and \ket{a} is a fixed unit vector in \A.
Again, we abbreviate $\r\ox \a$ to \rz.
Let $U$ be a unitary operator on \H.
\end{setting}

\begin{definition}\label{d:det2}
\T\ \emph{computes $U$ on principal input \r\ along branch \b} if the operator $\Tr_\A(C_\b \rz C_\b\dg)$ agrees with $U\r\,U\dg$ up to a scalar factor and thus represents the same state. \qef
\end{definition}

\begin{corollary}\label{c:det2}
If \T\ computes $U$ on principal input \r\ along branch \b\ then\\ $\Tr_\A(C_\b \rz C_\b\dg) = P_\T(\b|\rz)\cdot U\r\,U\dg$.
\end{corollary}

\begin{proof}
According to the definition of ``computes,'' the two sides of the claimed equation differ by only a positive scalar factor.
So it suffices to check that they have the same trace.
Since $\Tr(\rz) = \Tr(\r)$ and, by unitarity of $U$, $\Tr(U\r\,U\dg) = \Tr(\r)$, we have
$
\Tr\left(\Tr_\A(C_\b \rz C_\b\dg)\right)
= \Tr\left(C_\b \rz C_\b\dg\right) = \pbrz \Tr(\rz)
= \Tr\Big(\pbrz U\r\,U\dg\Big).
$
\end{proof}

\section{Input independence}
\label{s:ii}

We use Setting~\ref{s2}.

\begin{theorem}\label{t:iit}
Suppose that a meas tree \T\ computes a unitary operator $U$ on all
pure inputs along a branch \b. Then there is a vector $b$ in \A\
such that  for all (pure or not) principal inputs $\rho\in\DO{\H}$, we have:
\begin{enumerate}
\item $C_\b \rz C_\b\dg = U\rho\, U\dg \ox \ketbra{b}{b}$,
\item the probability \pbrz\ that a computation of\,\ \T\ follows branch $\b$ on input \rz\ is $\Tr(\ketbra{b}{b}) = \Vert b\Vert^2$ and thus is independent of $\rho$, and
\item if $\b$ is attainable on some input, then it is attainable on all inputs.
\end{enumerate}
\end{theorem}

\begin{proof}
Claim~3 obviously follows from claim~2. Claim~2 follows from claim~1:
\[ \pbrz =
\frac{\Tr(C_\mu \rz C_\mu\dg)}{\Tr(\rz)}
= \frac{\Tr(U\rho\, U\dg \ox \ketbra{b}{b})}{\Tr(\rz)}
= \frac{\Tr(\rho)\times\Tr(\ketbra{b}{b})}{\Tr(\rz)}
= \Tr(\ketbra{b}{b}). \]
In the rest of the proof we prove Claim~1.

By the linearity of the equation in Claim~1, we may assume without
loss of generality that $\rho$ is pure. Then $U\rho\, U\dg$ is also
pure.
Since \T\ computes $U$ on $\rho$, $\Tr_\A (C_\b \rz C_\b\dg)$ and $U\rho\, U\dg$ agree up to a scalar factor.
By the Pure State Factor Theorem\footnote{%
In \cite{Ballentine} Ballentine works with normalized density operators, but the theorem remains true, because any positive scalar factor can be shifted to $\xi(\rho)$.}
in \cite[\S8.3]{Ballentine},
\[ C_\b \rz C_\b\dg = U\rho\, U\dg \ox \xi(\rho) \]
for some (possibly mixed) state $\xi(\rho)$ in \A,
and therefore $\Tr(C_\b \rz C_\b\dg)
= \Tr(\rho)\times \Tr(\xi(\rho))$. Notice that a mixed state $\sigma$ is pure if and only if $\Tr(\sigma^2) = \Tr(\sigma)^2$. Since $\rho$ is pure, the states \rz\ and $C_\b\rz C_\b\dg$ are pure. We have
\begin{align*}
\left(\Tr(\rho)\right)^2
  \times \left(\Tr(\xi(\rho))\right)^2
&= \big(\Tr(C_\b \rz C_\b\dg)\big)^2\\
&= \Tr\big((C_\b \rz C_\b\dg)^2\big)
= \Tr(\rho^2)
  \times \Tr\big((\xi(\rho))^2\big).
\end{align*}
Cancelling $\Tr(\rho)^2 = \Tr(\rho^2)$, we get $\left(\Tr(\xi(\rho))\right)^2 = \Tr\left(\xi(\rho)^2\right)$, and so $\xi(\rho)$ is pure as well.

Let $\rho = \ketbra\psi\psi$.
Then there is an \A\ vector \ket{f(\psi)} such that
$\xi(\rho) = \ket{f(\psi)}\bra{f(\psi)}$.
Thus, $C_\b \rz C_\b\dg = U\rho\, U\dg \ox \ket{f(\psi)}\bra{f(\psi)}$. Hence
$C_\b (\ket\psi\ox\ket{a}) = U\ket\psi\ox \ket{f(\psi)}$
up to a phase, which we absorb into $\ket{f(\psi)}$.

Since $U\ket\psi\ox \ket{f(\psi)} = C_\b (\ket\psi\ox\ket{a})$ is linear in \ket\psi, we can apply the function constancy criterion Theorem~\ref{t:const} in Appendix~A with $V_1=\H$, $V_2 = \H$, $V_3=\A$, $L = U$, and $f$ is the function $f(\psi)$.
We learn that function $f(\psi)$ is constant, with the same value \ket{b} on all nonzero vectors $\psi$ in $\H$.
So we have
$C_\b\big(\ket\psi\ox\ket{a}\big)
 = U\ket\psi\otimes\ket{b}$
for every vector \ket\psi\ in $\H$.
In terms of density operators,
$C_\b \rz C_\b\dg
= U\rho\, U\dg \ox \ketbra{b}{b}$
for all pure $\rho\in\DO{\H}$ and, by linearity, for all $\rho\in\DO{\H}$.
\end{proof}

\begin{corollary}\label{c:puretoall}
If \T\ computes $U$ on pure inputs along a branch, then it computes $U$ on all inputs along that branch.
\end{corollary}

\begin{proof}[Proof of Theorem~\ref{t:ii}]
The reduction theorem, Theorem~\ref{t:reduction}, allows us to transfer the results of Theorem~\ref{t:iit} from meas trees to circuits.
Thus Theorem~\ref{t:ii} follows from the reduction theorem and Theorem~\ref{t:iit}.
\end{proof}

\begin{appendices}

\section{A function constancy criterion}

Let $V_1, V_2, V_3$ be vector spaces, $L: V_1\to V_2$ a linear
transformation of rank $\ge2$, $f$ an arbitrary function from $V_1$ to $V_3$, and $\L = L\ox f$ meaning that $\L x = Lx \ox fx$.

Notice that, if $f$ is constant, then \L\ is linear, because $L$ is linear and \ox\ is bilinear.
The same conclusion follows if $f$ is constant only on those vectors $x$ in $V_1$ for which $Lx\ne0$.
Indeed, changing the value of $fx$ arbitrarily for those $x$ with $Lx=0$ has no effect on \L.
The following theorem provides a useful converse to this observation.

\begin{theorem}\label{t:const}
Suppose that \L\ is linear.
Then $fx = fy$ for all $x,y$ with nonzero $Lx, Ly$.
\end{theorem}

\begin{proof}
It suffices to prove $fx=fy$ when $Lx,Ly$ are linearly independent. Indeed, there are vectors $v_1, v_2$ such that $Lv_1, Lv_2$ are independent because the rank of $L$ is $\ge2$. If $Lx,Ly$ are linearly dependent but nonzero, then some $Lv_i$ is independent from both $Lx$ and $Ly$. Then $fx = fv_i = fy$.

So suppose that $Lx, Ly$ are independent, and choose bases
\begin{itemize}
\item $v_1, v_2, \dots$ for $V_2$, where $v_1 = Lx$ and $v_2 = Ly$, and
\item $w_1, w_2, \dots$ for $V_3$, where  $w_1 = fx$ and $fy = aw_1 + bw_2$ for some scalars $a,b$.
\end{itemize}
Then $f(x+y) = \sum_j c_j w_j$ for some scalars $c_j$.
Notice that $V_2$ and $V_3$ are arbitrary vector spaces and  need not have any inner product specified.
Even if they (or some of them) are infinite dimensional, we are talking about bases in the simple sense of linear algebra.
Therefore $\sum_j c_j w_j$ is a finite sum, i.e., all but finitely many of the coefficients $c_j$ are 0.

By linearity of $L$ and \L, we have
\begin{align*}
& (v_1 + v_2) \ox \sum_j c_j w_j = (Lx + Ly) \ox f(x+y)
 = L(x+y) \ox f(x+y) = \\
& \L(x+y) = \L x + \L y
 = (Lx \ox fx) + (Ly \ox fy) =
 v_1\ox w_1 + v_2\ox(a w_1 + bw_2).
\end{align*}
Vectors $v_i\ox w_j$ form a basis of $V_2\ox V_3$.
Comparing the coefficients of basis vectors\\
$v_1\ox w_1, v_1\ox w_2, v_2\ox w_1,  v_2\ox w_2$ at the left end and the right end, we have
\[ c_1 = 1,\quad c_2 = 0,\quad c_1 = a = 1,\quad c_2 = b = 0, \]
and therefore
$ fy = 1w_1 + 0w_2 = w_1 = fx $.
\end{proof}

\section{Isometries}

A linear operator $L$ on a Hilbert space \K\ is an \emph{isometry} if it preserves norms: $\norm{L\ket\psi} = \norm{\ket\psi}$ for every vector \ket\psi\ in \K\ \cite[Def.~2.10]{Moretti}.

Equivalently, $L$ preserves inner products \cite[Prop.~3.8]{Moretti}.
It follows that a (bounded) linear operator $L$ is an isometry if and only if $L\dg L = \Id$:
\[
\braket{L\psi}{L\phi} = \braket\psi\phi\text{ for all } \ket\psi, \ket\phi
\iff \braket{L\dg L\psi}\phi = \braket{\Id\psi}\phi\text{ for all } \ket\psi, \ket\phi
\iff L\dg L = \Id.
\]
Notice that a single-outcome measurement consists of a single linear operator $L$ with $L\dg L = \Id$ and therefore amounts to an isometry.

If $L$ is an isometry then
$\Tr(L\r L\dg) = \Tr(L\dg L\r) = \Tr(\r)$ for all \r\ in DO(\K).
In finite dimensions $L\dg L = \Id$ implies that $L L\dg = \Id$, so that $L\dg = L^{-1}$ and $L$ is unitary.

\begin{theorem}\label{t:iit2}
The input independence theorem, Theorem~\ref{t:ii}, remains valid if ``a unitary operator'' is replaced by ``an isometry''.
\end{theorem}

\begin{proof}
By the reduction theorem, Theorem~\ref{t:reduction}, it suffices to prove the version of Theorem~\ref{t:iit} where ``a unitary operator'' is replaced by ``an isometry''.
We walk through the proof of Theorem~\ref{t:iit} and examine all places where the unitarity is or seems to be used.
\begin{itemize}
\item The derivation of Claim~2 from Claim~1 uses only the equality $\Tr(U\r\,U\dg) = \Tr(\r)$ which is true for isometries $U$.

\item The application of Pure State Factor Theorem remains valid because, if \r\ is a pure density operator \ketbra\psi\psi, then  $U\rho U\dg = U\ketbra\psi\psi U\dg = \ketbra{U\psi}{U\psi}$ is pure as well for any linear operator $U$.

\item The proof that $\xi(\r)$ is pure requires only a tiny modification:
\begin{multline*}
\Tr\big((U\r\,U\dg)^2\big) \times \left(\Tr(\xi(\rho))\right)^2
= \big(\Tr(U\r\,U\dg)\big)^2 \times \left(\Tr(\xi(\rho))\right)^2
= \big(\Tr(C_\b \rz C_\b\dg)\big)^2\\
= \Tr\big((C_\b \rz C_\b\dg)^2\big)
= \Tr\big((U\r\,U\dg)^2\big) \times \Tr\big((\xi(\rho))^2\big).
\end{multline*}
Cancelling $\Tr\big((U\r\,U\dg)^2\big)$, we get $\left(\Tr(\xi(\rho))\right)^2 = \Tr\left(\xi(\rho)^2\right)$, and so $\xi(\rho)$ is pure.
\end{itemize}
Thus, only  the isometricity of $U$ is used in the proof of Theorem~\ref{t:iit}.
\end{proof}

The following proposition shows that the isometricity requirement in Theorem~\ref{t:iit2} cannot be substantially weakened at least in the case where the same linear operator is computed along every computation path.

\begin{proposition}
If a quantum circuit computes the same linear operator $U$ on \H\ along every computation path, then $U$ is an isometry up to a positive scalar factor.
\end{proposition}

\begin{proof}
By the reduction theorem, Theorem~\ref{t:reduction}, it suffices to prove the meas tree version of the proposition: If a meas tree computes the same linear operator $U$ along every branch \b, then $U$ is an isometry.

By claim~1 of Theorem~\ref{t:iit2}, for every branch \b\ there is a vector $b_\b$ in \A\ such that
$C_\b \rz C_\b\dg = U\rho\, U\dg \ox \ketbra{b_\b}{b_\b}$.
It follows that
\[
 \Tr_\A\sum_\b(C_\b \rz C_\b^\dag)
 = U\r\,U\dg\ox \sum_\b \Tr\ketbra{b_\b}{b_\b}
 = t^2\cdot U\r\,U\dg,
\]
where $t = \sqrt{\sum_\mu \braket{b_\mu}{b_\mu}}$. We show that $tU$ is isometric.

Let \ket\psi\ be an arbitrary vector in \K\ and $\r = \ketbra\psi\psi$.
By Lemma~\ref{l:cum}, $\sum_\b C_\b\dg C_\b = \Id$. We have:
\begin{align*}
\langle tU\psi | tU\psi\rangle
& = \Tr(t^2\cdot U\r\,U\dg)
  =\Tr\Big(\Tr_\A\sum_\b(C_\b \rz C_\b\dg)\Big)\\
& = \sum_\b \Tr(C_\b\dg C_\b \rz)
= \Tr(\rz) = \Tr(\rho)
=\braket\psi\psi. \qedhere
\end{align*}
\end{proof}

\begin{corollary}
Suppose that the Hilbert space \H\ is finite dimensional.
If a quantum circuit computes the same linear operator $U$ on \H\ along every computation path, then $U$ is unitary up to a positive scalar factor.
\end{corollary}

\end{appendices}

\end{document}